\documentclass[12pt]{article}

\usepackage[utf8]{inputenc}
\usepackage{fullpage}
\usepackage[english]{babel}
\usepackage{amsmath}
\usepackage{amssymb}
\usepackage{amsthm}
\usepackage[hidelinks]{hyperref}

\newtheorem{proposition}{Proposition}
\theoremstyle{definition}
\newtheorem{definition}{Definition}
\newtheorem*{notation}{Notation}
\newcommand{\hole}{{[\phantom x]}}
\newcommand{\atom}[1]{{``#1"}}
\newcommand{\fv}{{\operatorname{FV}}}
\newcommand{\rb}{{\operatorname{RB}}}
\newcommand{\lamrb}{{\Lambda^\rb}}
\newcommand{\bnf}{{\operatorname{NF}}}

\title{Lambda Calculus with Explicit Read-back}
\author{Anton Salikhmetov}

\begin{document}
\maketitle

\begin{abstract}
This paper introduces a new term rewriting system that is similar to the embedded read-back mechanism for interaction nets presented in our previous work, but is easier to follow than in the original setting and thus to analyze its properties.
Namely, we verify that it correctly represents the lambda calculus.
Further, we show that there is exactly one reduction sequence that starts with any term in our term rewriting system.
Finally, we represent the leftmost strategy which is known to be normalizing.
\end{abstract}

\section{Introduction}

Read-back is the process of decoding a $\lambda$-term from its another representation.
Previously, \cite[Section 7]{termgraph} has presented an embedded read-back mechanism for interaction nets.
Here, we will define a novel term rewriting system which resembles that mechanism in order to check its correctness and to study its properties more easily than in the original setting.

\begin{notation}
We will use notations similar to \cite{barendregt}.
$C\hole$ is a context with one hole, which will be denoted with $\hole$.
$C[M]$ is the result of placing $M$ in the hole of the context $C\hole$.
$\Lambda$ is the set of all $\lambda$-terms.
$\bnf$ is the set of all $\lambda$-terms in $\beta$-normal form.
$\fv(M)$ is the set of all free variables in a $\lambda$-term $M$.
$M[x:=N]$ is the result of substituting $N$ for all free occurrences of variable $x$ in $M$.
$M\ \vec N \equiv M\ N_1\cdots N_n$, where $n \ge 0$.
If $M, N \in \Lambda$, then $M \rightarrow N$ stands for $\beta$-reduction, and we write $M \rightarrow_= N$ when $M \equiv N$ or $M \rightarrow N$.
\end{notation}

\begin{definition}
\label{bulletdef}
Let $M \in \Lambda$.
$\atom M$ is an \textit{atom}.
$M^\bullet$ is the result of replacing each $x \in \fv(M)$ in $M$ with $\atom x$.
Additionally, we define set $\Lambda^\bullet = \{M^\bullet\ |\ M \in \Lambda\}$, then extend the definitions of substitution and $\fv$ for $\Lambda^\bullet$ according to $\fv(\atom x) = \varnothing$ for any variable $x$.
\end{definition}

Note that for any $M, N \in \Lambda^\bullet$, we have $\fv(M) = \varnothing$ and $M[x:=N] \equiv M$.

\begin{definition}
$\lamrb$ is the minimal set that satisfies the following conditions:
\begin{align*}
\forall M \in \Lambda&:
\ \atom M \in \lamrb; \\
\forall C\hole,\ M \in \Lambda^\bullet&:
\ \langle C\hole,\ M \rangle \in \lamrb; \\
\forall C\hole,\ M \in \lamrb,\ \vec N \in \Lambda^\bullet&:
\ \langle C\hole,\ M\ \vec N \rangle \in \lamrb.
\end{align*}
\end{definition}

\begin{definition}
\label{rbred}
Let us introduce a reduction relation $\rightarrow$ on $\lamrb$ as follows:
\begin{align*}
\forall C\hole,\ M \in \Lambda&:
\ \langle C\hole,\ \atom M \rangle \rightarrow \atom{C[M]}; \\
\forall C\hole,\ \lambda x.M \in \Lambda^\bullet&:
\ \langle C\hole,\ \lambda x.M \rangle \rightarrow \langle C[\lambda x.\hole],\ M[x:=\atom x] \rangle; \\
\forall C\hole,\ M \in \Lambda,\ N_0, \vec N \in \Lambda^\bullet&:
\ \langle C\hole,\ \atom M\ N_0\ \vec N \rangle \rightarrow \langle C\hole,\ \langle M\ \hole,\ N_0 \rangle\ \vec N \rangle; \\
\forall C\hole,\ \lambda x.M, N_0, \vec N \in \Lambda^\bullet&:
\ \langle C\hole,\ (\lambda x.M)\ N_0\ \vec N \rangle \rightarrow \langle C\hole,\ M[x:=N_0]\ \vec N \rangle; \\
\forall M, M' \in \lamrb,\ C\hole,\ \vec N \in \Lambda^\bullet&:
\ M \rightarrow M'\ \Rightarrow\ \langle C\hole,\ M\ \vec N \rangle \rightarrow \langle C\hole,\ M'\ \vec N \rangle.
\end{align*}
\end{definition}

Thanks to $\Lambda \cap \lamrb = \varnothing$, the introduced relation is not to be confused with $\beta$-reduction.

\section{Correctness}

\begin{definition}
\label{readback}
Mapping $\rb: \Lambda^\bullet \cup \lamrb \rightarrow \Lambda$ is called \textit{read-back} and defined as follows:
\begin{align*}
\forall C\hole,\ M \in \Lambda^\bullet \cup \lamrb,\ \vec N \in \Lambda^\bullet&:
\ \rb(\langle C\hole,\ M\ \vec N\rangle) \equiv C[\rb(M)\ \rb(N_1) \cdots \rb(N_n)]; \\
\forall M, N \in \Lambda^\bullet&:
\ \rb(M\ N) \equiv \rb(M)\ \rb(N); \\
\forall \lambda x.M \in \Lambda^\bullet&:
\ \rb(\lambda x.M) \equiv \lambda x.\rb(M[x:=\atom x]); \\
\forall M \in \Lambda&:
\ \rb(\atom M) \equiv M.
\end{align*}
\end{definition}

For example, if $M$ is a $\lambda$-term, then $\rb(\langle \hole,\ M^\bullet\rangle) \equiv \rb(M^\bullet) \equiv M$.

\begin{proposition}
\label{betarb}
$\forall \lambda x.M, N \in \Lambda^\bullet:\ \rb((\lambda x.M)\ N) \rightarrow \rb(M[x:=N])$.
\end{proposition}
\begin{proof}
Let $M' \equiv \rb(M[x:=\atom x])$.
Notice that $\rb(M[x:=N]) \equiv M'[x:=\rb(N)]$.
Then we have $\rb((\lambda x.M)\ N) \equiv (\lambda x.M')\ \rb(N) \rightarrow M'[x:=\rb(N)] \equiv \rb(M[x:=N])$.
\end{proof}

\begin{proposition}
\label{rbcorrect}
$\forall L, R \in \lamrb:\ L \rightarrow R\ \Rightarrow\ \rb(L) \rightarrow_= \rb(R)$.
\end{proposition}
\begin{proof}
We will use induction on Definition~\ref{rbred}.
First, consider the four basic cases:
\begin{enumerate}
\item $L \equiv \langle C\hole,\ \atom M\rangle \rightarrow \atom{C[M]} \equiv R$. \\
Notice $\rb(L) \equiv C[M] \equiv \rb(R)$.

\item $L \equiv \langle C\hole,\ \lambda x.M \rangle \rightarrow \langle C[\lambda x.\hole],\ M[x:=\atom x] \rangle \equiv R$. \\
Notice $\rb(L) \equiv C[\lambda x.\rb(M[x:=\atom x])] \equiv \rb(R)$.

\item $L \equiv \langle C\hole,\ \atom M\ N_0\ \vec N \rangle \rightarrow \langle C\hole,\ \langle M\ \hole,\ N_0 \rangle\ \vec N \rangle \equiv R$. \\
Notice $\rb(L) \equiv C[M\ \rb(N_0) \cdots \rb(N_n)] \equiv \rb(R)$.

\item $L \equiv \langle C\hole,\ (\lambda x.M)\ N_0\ \vec N \rangle \rightarrow \langle C\hole,\ M[x:=N_0]\ \vec N \rangle \equiv R$. \\
Notice $\rb(L) \equiv C[\rb((\lambda x.M)\ N_0)\ \rb(N_1) \cdots \rb(N_n)]$, and \\
\phantom{Notice} $\rb(R) \equiv C[\rb(M[x := N_0])\ \rb(N_1) \cdots \rb(N_n)]$. \\
Then $\rb(L) \rightarrow \rb(R)$ due to Proposition~\ref{betarb}.
\end{enumerate}
Now consider $L \equiv \langle C\hole,\ M\ \vec N \rangle \rightarrow \langle C\hole,\ M'\ \vec N \rangle \equiv R$, where $M \rightarrow M'$, then
notice that $\rb(L) \equiv C[\rb(M)\ \rb(N_1) \cdots \rb(N_n)]$ and $\rb(R) \equiv C[\rb(M')\ \rb(N_1) \cdots \rb(N_n)]$.
Finally, assuming that $\rb(M) \rightarrow_= \rb(M')$, conclude with $\rb(L) \rightarrow_= \rb(R)$.
\end{proof}

\section{Normalization}

\begin{definition}
\label{normalcontext}
A context $C\hole$ is \textit{normal} if and only if $\forall M \in \bnf:\ C[M] \in \bnf$.
\end{definition}

In particular, contexts $\hole$ and $\lambda x.\hole$ are both normal, however $\hole\ M$ is not.
Another example is context $C\hole \equiv x\ \vec N\ \hole$: if $x$ is a variable and $\vec N \in \bnf$, then $C\hole$ is also normal.

\begin{notation}
$[\bnf]$ will denote the set of all normal contexts.
\end{notation}

\begin{proposition}
\label{ctxcomp}
$\forall C_1\hole, C_2\hole \in [\bnf]:\ C_1[C_2\hole] \in \bnf$.
\end{proposition}
\begin{proof}
Let $M \in \bnf$.
Then by Definition~\ref{normalcontext}, $C_2[M] \in \bnf$, and $C_1[C_2[M]] \in \bnf$.
\end{proof}

\begin{proposition}
\label{nfatom}
Let $M \in \lamrb$.
Then $M$ is in normal form if and only if $M$ is an atom.
\end{proposition}
\begin{proof}
Notice that in Definition~\ref{rbred} we have $M \rightarrow N$ if and only if $M$ is not an atom.
\end{proof}

\begin{proposition}
\label{onlyseq}
For any $M \in \lamrb$, at most one reduction sequence starts from $M$.
\end{proposition}
\begin{proof}
Notice that for any $M \in \lamrb$ there is at most one $N \in \lamrb$ such that $M \rightarrow N$.
\end{proof}

\begin{proposition}
\label{rbnormal}
$\forall C\hole \in [\bnf],\ M, N \in \Lambda:\ \langle C\hole,\ M^\bullet \rangle \rightarrow^* \atom N\ \Rightarrow\ N \in \bnf$.
\end{proposition}
\begin{proof}
Let $L \equiv \langle C\hole,\ M^\bullet \rangle \rightarrow \atom N$.
Consider the following four cases:
\begin{enumerate}
\item $M \equiv x$, where $x$ is a variable. \\
Then $L \equiv \langle C\hole,\ \atom x \rangle \rightarrow \atom{C[x]}$ and $C[x] \in \bnf$.

\item $M \equiv \lambda x.M_0$, where $M_0 \in \Lambda$. \\
Then $L \rightarrow \langle C[\lambda x.\hole],\ M^\bullet_0\rangle$. \\
Notice $C[\lambda x.\hole] \in [\bnf]$ due to $\lambda x.\hole \in [\bnf]$ and Proposition~\ref{ctxcomp}.

\item $M \equiv (\lambda x.M_0)\ N_0\ \vec N$, where $M_0, N_0, \vec N \in \Lambda$. \\
Then $L \rightarrow \langle C\hole,\ (M_0[x:=N_0]\ \vec N)^\bullet \rangle$.

\item $M \equiv x\ N_0\ \vec N$, where $x$ is a variable and $N_0, \vec N \in \Lambda$. \\
Then $L \equiv \langle C\hole,\ \atom x\ N^\bullet_0 \cdots N^\bullet_n \rangle \rightarrow \langle C\hole,\ \langle x\ \hole,\ N^\bullet_0 \rangle\ N^\bullet_1 \cdots N^\bullet_n \rangle$. \\
Notice that $\forall \vec Q \in \bnf:\ x\ \vec Q\ \hole \in [\bnf]$. \\
Since $L \rightarrow^* \atom N$, we have $\langle x\ \hole,\ N^\bullet_0 \rangle \rightarrow^* \atom{P}$ for some $P \in \Lambda$. \\
Assume $P \in \bnf$ and use induction to conclude $N \in \bnf$.
\end{enumerate}
Due to Proposition~\ref{nfatom} and Proposition~\ref{onlyseq}, the above four cases are exhaustive.
\end{proof}

\begin{proposition}
\label{normalization}
$\forall M, N \in \Lambda:\ \langle \hole,\ M^\bullet \rangle \rightarrow^* \atom N\ \Leftrightarrow\ M \rightarrow^* N\ \wedge\ N \in \bnf$.
\end{proposition}
\begin{proof}
Notice that any infinite reduction sequence on $\lamrb$ can only be due to $\beta$-reduction.
Further, the proof of Proposition~\ref{rbnormal} shows that the leftmost $\beta$-redex is always contracted, thus we can use \cite[Normalization theorem 13.2.2]{barendregt} to conclude ($\Leftarrow$).

Conversely, ($\Rightarrow$) directly follows from Proposition~\ref{rbcorrect} and Proposition~\ref{rbnormal}.
\end{proof}

\section{Conclusion}

This paper introduced a new term rewriting system designed after the embedded read-back mechanism for interaction nets that was presented in our previous work \cite[Section 7]{termgraph}.

Then we have demonstrated the correctness of our term rewriting system and showed its normalization property in Proposition~\ref{normalization}.
As a simple corollary, the conjecture in \cite{mlc} can thus be proven using Proposition~\ref{normalization} and the preceding work~\cite{sinot}.
Proposition~\ref{normalization} can also help investigate the similar conjectures in \cite[Section 7]{termgraph} and \cite[Section 2]{impure}.

In the future, we intend to apply a similar technique to develop the ideas from \cite{impure} which we believe to be a very promising direction, especially taking into account the experimental results obtained from their software implementation\footnote{\url{https://www.npmjs.com/package/@alexo/lambda}}.

\end{document}